\documentclass[11pt]{article}
\usepackage{amsmath,amssymb,bm,amsthm,array,tikz,enumerate}
\usepackage{url}
\usepackage{lscape}
\usepackage[margin=30truemm]{geometry}
\usepackage{float}
\usepackage{makecell}
\usepackage{blkarray}
\usepackage{multirow}

\usetikzlibrary{positioning,graphs}

\numberwithin{equation}{section}

\DeclareMathOperator{\Spec}{Spec}
\DeclareMathOperator{\rank}{rank}
\DeclareMathOperator{\tr}{tr}
\DeclareMathOperator{\diag}{diag}

\DeclareMathOperator{\Aut}{Aut}

\newcommand{\MC}[1]{\mathcal{#1}}
\newcommand{\MB}[1]{\mathbb{#1}}

\newcommand{\BM}[1]{{\bm #1}}

\newtheorem{thm}{Theorem}[section]
\newtheorem{lem}[thm]{Lemma}
\newtheorem{pro}[thm]{Proposition}
\newtheorem{cor}[thm]{Corollary}

\theoremstyle{definition}

\allowdisplaybreaks

\title
{
Entanglement entropy in two-particle Grover walks on graphs \\
}

\author{
Sho Kubota\thanks{
Corresponding author.
Department of Mathematics Education,
Aichi University of Education,
1 Hirosawa, Igaya-cho, Kariya, Aichi 448-8542, Japan.
\texttt{skubota@auecc.aichi-edu.ac.jp}}
\and
Haruhiko Matsubara\thanks{
Graduate School of Environment Information Sciences,
Yokohama National University,
Hodogaya, Yokohama 240-8501, Japan.
}
\and
Etsuo Segawa\thanks{
Graduate School of Environment Information Sciences,
Yokohama National University,
Hodogaya, Yokohama 240-8501, Japan.
}
}
\date{}

\begin{document}
\maketitle

\begin{abstract}
We define a two-particle quantum walk of identical particles on a graph $G$
via the one-particle Grover walk on the Kronecker product $G \otimes G$,
and call it the two-particle Grover walk.
In systems of identical particles,
quantum mechanics requires that quantum states have a certain invariance with respect to the exchange of particles.
Focusing on the symmetry of the Kronecker product $G \otimes G$ as a graph,
we show that the time evolution operator of this walk commutes with the swap operator,
which ensures that this requirement is satisfied.
Furthermore, we study the entanglement entropy of quantum states evolved by this walk.
For the complete bipartite graph $K_{n,n}$,
we completely determine the values of $n$ for which the quantum states evolved from specific initial states attain the upper bound of the entropy at some time,
and prove that they are exactly $1$ and $2$.
\vspace{8pt} 
\end{abstract}

\section{Introduction}

Quantum walks are concepts that incorporate quantum-mechanical features into classical random walks.
Although they originate from stochastic processes,
they are now also studied as important objects in quantum information science, including quantum algorithms~\cite{portugal2013quantum} and quantum simulation~\cite{qiang2024quantum}.
It is also known that discrete-time quantum walks can realize universal quantum computation~\cite{lovett2010universal},
and the study of quantum walks is related to the theoretical understanding of quantum information and quantum computation.

Many studies of quantum walks deal with one-particle models.
In such models, the state of a single particle is represented as an element of a Hilbert space,
and its time evolution is studied.
On the other hand,
the model treated in this paper is a discrete-time two-particle model.
In this model, the state of two particles is represented as an element of the tensor product of the Hilbert spaces describing the individual particles, and its time evolution is considered.
In two-particle systems, it is necessary to consider issues that do not arise in one-particle systems,
such as the distinguishability of particles and the related exchange symmetry.
At the same time,
phenomena that do not arise in one-particle systems also appear. 
The most representative example is the entanglement between particles.
Although this paper focuses on the two-particle case,
quantum walks with multiple particles have also been studied on arbitrary graphs~\cite{rohde2011multi} and on one-dimensional lattices~\cite{chandrashekar2012quantum, goyal2010spatial}.
Moreover, Rohde et al.~\cite{rohde2013increasing} showed that multi-walker quantum walks on a line can be mapped to single-walker quantum walks on higher-dimensional lattice graphs.
Typical research themes in two-particle quantum walks on one-dimensional lattices include how interactions between particles affect the asymptotic behavior for large time steps and how entanglement measures evolve over time~\cite{ahlbrecht2012molecular,  carson2015entanglement,omar2006quantum}.
On the other hand, on finite graphs, two-particle quantum walks driven by a temporally random unitary time evolution with respect to the connection to each edge on the cycle and line have also been studied,
and characterized the attractor space where the two walkers are absorbed in the long time limit~\cite{paryzkova2024two}.
Furthermore, there are also studies applying two-particle quantum walks to the graph isomorphism problem~\cite{berry2011two, gamble2010two, rudinger2013comparing}.
Recent studies on quantum walks of identical particles have investigated particle statistics, correlations, and entanglement in different settings~\cite{camacho2025quantum, mastandrea2026spatial}.

In this paper, we treat discrete-time two-particle quantum walks of identical particles on graphs.
One of the major differences from previous studies is the treatment of interactions.
A typical way to introduce interactions in two-particle quantum walks is to add local interactions when the particles are at the same position or sufficiently close to each other.
On the other hand, our proposed two-particle quantum walk is based on a natural isomorphism between the ``two-particle Hilbert space on a graph $G$" and ``the one-particle Hilbert space on the Kronecker product graph $G \otimes G$".
Through this isomorphism, we will show that
a one-particle quantum walk on $G\otimes G$ reproduces a two-particle quantum walk model on $G$ with some interaction characterized by quantum coins assigned at each vertex of $G\otimes G$~(see Lemma~\ref{0422-1}).
Then in this paper, we assign the Grover matrix at each vertex of $G\otimes G$.
We call this walk the {\it two-particle Grover walk on $G$}.
Note that the two-particle Grover walk on $G$ coincides with the one-particle Grover walk on $G \otimes G$.
The detailed definition is given in Section~\ref{0527-1}. 
This construction yields a time evolution having global interactions between two particles that differs both from the time evolutions of models with local interactions and from the non-interacting time evolution obtained by taking the tensor product of two one-particle time evolutions.
Thus, our model can be regarded as a natural generalization, from a graph theoretic viewpoint, of the idea of Rohde et al.~\cite{rohde2013increasing} from a line to an arbitrary graph $G$.
In our construction, the corresponding one-particle walk is defined on the Kronecker product graph $G\otimes G$, which corresponds to what they call the virtual graph.
While Rohde et al.~\cite{rohde2011multi} have already considered multi-walker quantum walks on arbitrary graphs, our construction is distinguished by the choice of the Grover coin on this product graph.

In addition, when considering two-particle systems of identical particles,
we must ensure that the time evolution operator satisfies the requirement of quantum mechanics that it preserves symmetry with respect to the exchange of particles~\cite{sakurai2020modern}.
This implies that the time evolution operator must commute with the swap operator exchanging the locations of the two particles.
To satisfy this requirement,
we focus on the symmetry of $G \otimes G$ as a graph and prove this commutativity by using methods from algebraic graph theory.
Although the definition of our two-particle quantum walk is given in Section~\ref{0527-1},
we state this result here as the first main theorem of this paper.

\begin{thm}\label{thm:main1}
Let $\MC{U}$ be the time evolution operator of the two-particle Grover walk on a graph $G$,
and let $P$ be the swap operator.
Then, we have $\MC{U}P = P\MC{U}$.
\end{thm}

Therefore, once the graph $G$ is fixed, our construction uniquely specifies the time evolution operator of the two-particle Grover walk.
Theorem~\ref{thm:main1} ensures that this time evolution operator satisfies the indistinguishability requirement.
Now, we focus on the von Neumann entropy of the one-particle reduced density matrix for states in the two-particle Grover walks as a feature that can be obtained precisely because this is a multi-particle system.
In this paper, we refer to this quantity as the entanglement entropy.
The entanglement entropy is an indicator of how the two particles are entangled.
For a graph $G = (V,E)$, the one-particle Hilbert space in our model has dimension $D := |\MC{A}(G)| = 2|E|$.
It is well known that the entanglement entropy is bounded above by $\log D$~\cite{nielsen2010quantum}.
Thus, for the two-particle Grover walk on $G$, this upper bound is $\log(2|E|)$.
We note that if we had chosen the tensor product of one-particle time evolutions as the time evolution, then the entanglement entropy would be invariant to the time evolution, which means that 
the entanglement of this system is that of the initial state. 
We are interested in how large entanglement is generated from a small entanglement through the time evolution of the two-particle Grover walk on $G$.
In this paper, we set the initial state by choosing an edge $e\in E$ inducing the two symmetric arcs $a$ and $a^{-1}$ such that
\[
\psi_0^{(G,e)} = \frac{|a\rangle \otimes |a^{-1}\rangle \pm |a^{-1}\rangle \otimes |a\rangle}{\sqrt{2}},
\]
where the sign is either $+$ (boson) or $-$ (fermion).
Here $|z\rangle$ denotes the standard basis vector $\BM{e}_z$ corresponding to an arc $z$.
The entanglement entropy of $\psi_0^{(G,e)}$ is $\log 2$, which is independent of $D$. 
Let $\psi_t^{(G,e)}$ be the $t$-th iteration of the two-particle Grover walk on $G$ with the above initial state.
Natural questions arise that  
\begin{enumerate}
\item 
Is there a graph $G=(V,E)$ where there are a time $t\geq 0$ and an edge $e\in E$ such that the entanglement entropy of $\psi_t^{(G,e)}$ attains the upper bound $\log (2|E|)$?
\item 
If it exists, what is the family of such graphs? 
\end{enumerate}
In this paper, for complete bipartite graphs $K_{n,n}$ as a first trial to tackle these problems,
we determine the values of $n$ for which the entanglement entropy attains its upper bound at some time.

\begin{thm}\label{thm:main2}
For the two-particle Grover walk on $K_{n,n}$ starting from the initial states defined in~\eqref{0507-4},
the entanglement entropy attains its upper bound at some time if and only if $n = 1,2$.
Moreover, when $n=1$, the entanglement entropy attains its upper bound at every time, while when $n=2$, it attains its upper bound at time $\tau$ if and only if $\tau \equiv 2 \pmod{4}$.
\end{thm}

The remainder of this paper is organized as follows.
In Section~2, we prepare notation for graphs with algebraic graph theoretical lemmas and the definition of the Grover walk.
In Section~3, we provide the definition of the two-particle Grover walk on $G$ and the proof of Theorem~\ref{thm:main1} using an automorphism of the graph $G\otimes G$.
Section~4 is devoted to the proof of Theorem~\ref{thm:main2}.
Finally, we give a summary and discussion in Section~5.

\section{Preliminaries}

See \cite{godsil2013algebraic} for basic terminology related to graphs.
Let $G =(V, E)$ be a graph with the vertex set $V$ and the edge set $E$.
Throughout this paper, we assume that graphs are simple and finite,
i.e., $|V| < \infty$ and $E \subset \{\{x,y\} \subset V \mid x \neq y\}$.
We define the \emph{adjacency matrix} $A = A(G) \in \MB{C}^{V \times V}$ by
\[ A_{x,y} = \begin{cases}
1 \quad &\text{if $\{x,y\} \in E$,} \\
0 \quad &\text{otherwise.}
\end{cases} \]
The multiset of eigenvalues of the adjacency matrix of a graph is called the \emph{spectrum} of the graph and is denoted by $\Spec(G)$.
The \emph{complete bipartite graph} $K_{m,n}$ is the graph with vertex set and edge set given by
\begin{align*}
V(K_{m,n}) &:= \{ x_1, \dots, x_m, y_1, \dots, y_n \}, \\
E(K_{m,n}) &:= \{ \{x_i, y_j\} \mid i \in \{1, \dots,m \}, j \in \{1, \dots,n\} \}.
\end{align*}

The graph shown in Figure~\ref{0528-1} is $K_{2,3}$.
The corresponding adjacency matrix with respect to the vertex ordering $x_1,x_2,y_1,y_2,y_3$ is
\[ A(K_{2,3}) = \begin{bmatrix} 
0&0&1&1&1 \\ 
0&0&1&1&1 \\ 
1&1&0&0&0 \\ 
1&1&0&0&0 \\ 
1&1&0&0&0
\end{bmatrix}, \]
and computing its spectrum gives $\Spec(K_{2,3}) = \{ [\sqrt{6}]^1, [0]^3, [-\sqrt{6}]^1 \}$.
As shown here, the spectrum is written by enclosing the eigenvalues in square brackets and indicating their multiplicities as a superscript outside the brackets.

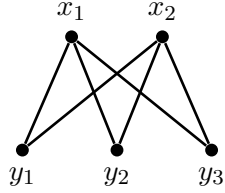
\begin{figure}[ht]
\centering
\begin{tikzpicture}
[scale = 0.5,
v/.style = {circle, fill = black, inner sep = 0.6mm},
u/.style = {circle, fill = white, inner sep = 0.1mm}
]
\node[v, label=above:$x_1$] (1) at (-1.2, 3) {};
\node[v, label=above:$x_2$] (2) at (1.2, 3) {};

\node[v, label=below:$y_1$] (3) at (-2.5, 0) {};
\node[v, label=below:$y_2$] (4) at (0, 0) {};
\node[v, label=below:$y_3$] (5) at (2.5, 0) {};

\draw[line width = 1pt] (1) to (3);
\draw[line width = 1pt] (1) to (4);
\draw[line width = 1pt] (1) to (5);

\draw[line width = 1pt] (2) to (3);
\draw[line width = 1pt] (2) to (4);
\draw[line width = 1pt] (2) to (5);
\end{tikzpicture}
\caption{The complete bipartite graph $K_{2,3}$} \label{0528-1}
\end{figure}

\subsection{One-particle Grover walks on graphs} \label{0519-1}

We first introduce the Grover walk on a graph.
The Grover walk discussed here is the Grover walk as a one-particle system.
This walk is also referred to as the arc-reversal walk or the arc-reversal Grover walk.
Let $G = (V, E)$ be a graph.
Define $\MC{A} = \MC{A}(G) := \{ (x, y), (y, x) \mid \{x, y\} \in E \}$,
which is the set of the \emph{symmetric arcs} of $G$.
The origin $x$ and terminus $y$ of $a=(x, y) \in \MC{A}$ are denoted by $o(a)$ and $t(a)$, respectively.
We write the inverse arc of $a$ as $a^{-1}$.
We regard the vector space $\MB{C}^{\MC{A}}$, whose coordinates are indexed by $\MC{A}$, as a Hilbert space with the standard inner product,
and denote its standard basis by $\{ \BM{e}_a \mid a \in \MC{A} \}$.
We call a vector $\psi \in \MB{C}^{\MC{A}}$ with $\| \psi \| = 1$ a \emph{quantum state}, or simply a \emph{state}.
The \emph{time evolution operator} $U = U(G) \in \MB{C}^{\MC{A} \times \MC{A}}$ for the Grover walk is defined by
\[ U_{a,b} := \frac{2}{\deg t(b)} \delta_{o(a), t(b)} - \delta_{a,b^{-1}} \]
for $a,b \in \MC{A}$, where $\delta$ is the Kronecker delta.
Although in this paper we have defined the time evolution operator $U$ by directly specifying its $(a,b)$-entry,
the Grover walk is a type of coined quantum walk,
and $U$ can be written as the product of a coin operator and a shift operator.
See~\cite{kubota2021quantum} for details on this formulation.
In our analysis, we will focus on how $U$ acts on the standard basis vectors:

\begin{lem} \label{0507-1}
For $a \in \MC{A}$, we have
\[ U\BM{e}_a = \frac{2}{\deg t(a)}\sum_{\substack{z \in \MC{A} \\ o(z) = t(a)}} \BM{e}_z - \BM{e}_{a^{-1}}. \]
\end{lem}

\begin{proof}
Indeed, for $w \in \MC{A}$,
\begin{align*}
(U\BM{e}_a)_w &= \sum_{z \in \MC{A}} U_{w,z}(\BM{e}_a)_z \\
&= \sum_{z \in \MC{A}} \left( \frac{2}{\deg t(z)} \delta_{o(w), t(z)} - \delta_{w, z^{-1}} \right) \delta_{a,z} \\
&= \frac{2}{\deg t(a)} \delta_{o(w), t(a)} - \delta_{w,a^{-1}}.
\end{align*}
On the other hand,
\[ \left( \frac{2}{\deg t(a)}\sum_{\substack{z \in \MC{A} \\ o(z) = t(a)}} \BM{e}_z - \BM{e}_{a^{-1}} \right)_w
= \frac{2}{\deg t(a)}\sum_{\substack{z \in \MC{A} \\ o(z) = t(a)}} \delta_{z,w} - \delta_{a^{-1}, w}. \]
For the sum on the right-hand side,
if $o(w) = t(a)$, then the arc $z$ satisfying $z = w$ automatically satisfies $o(z) = t(a)$.
Hence,
\[ \sum_{\substack{z \in \MC{A} \\ o(z) = t(a)}} \delta_{z,w} = 1. \]
Otherwise, the arc $z$ satisfying $z = w$ does not satisfy $o(z) = t(a)$.
Hence,
\[ \sum_{\substack{z \in \MC{A} \\ o(z) = t(a)}} \delta_{z,w} = 0. \]
In particular,
\[ \sum_{\substack{z \in \MC{A} \\ o(z) = t(a)}} \delta_{z,w} = \delta_{o(w), t(a)}. \]
Therefore, we obtain the desired equality.
\end{proof}

The action described in Lemma~\ref{0507-1} can be visualized as shown in Figure~\ref{0519-3}.
In this figure, the action of $U$ collects the arcs starting from the vertex $x$, where $x$ is the terminal vertex of the original arc.
Except for the reverse arc of the original one,
these arcs have weight $\frac{2}{\deg x}$,
while the reverse arc has weight $\frac{2}{\deg x} - 1$.

\begin{figure}[htb]
\begin{center}
\begin{tikzpicture}
[scale = 0.5,
v/.style = {circle, fill = black, inner sep = 0.6mm},
u/.style = {circle, fill = white, inner sep = 0.1mm}
]
\node[u] (x) at (0.7, 0) {$x$};
\node[u, white] (13) at (1.25, 2.7) {{\scriptsize$\frac{2}{\deg x}$}};
\node[u, white] (15) at (0, -2) {{\scriptsize$\frac{2}{\deg x} - 1$}};
\node[v] (1) at (0.5,0.5) {};
\node[v] (2) at (0.5, 2) {};
\node[v] (3) at (2, 2) {};
\node[v] (4) at (2, 0.5) {};
\draw[line width = 1pt] (1) to (2);
\draw[line width = 1pt] (1) to (3);
\draw[line width = 1pt] (1) to (4);
\node[v] (5) at (-0.5,-0.5) {};
\node[v] (6) at (-0.5, -2) {};
\node[v] (7) at (-2, -2) {};
\node[v] (8) at (-2, -0.5) {};
\draw[line width = 1pt] (1) to (5);
\draw[line width = 1pt] (5) to (6);
\draw[line width = 1pt] (5) to (7);
\draw[line width = 1pt] (5) to (8);
\node[u] (51L) at (-0.5, -0.2) {};
\node[u] (51R) at (0.2, 0.5) {};
\draw[draw = blue, line width = 1pt, ->] (51L) to (51R);
\end{tikzpicture}
\raisebox{45pt}{$\quad \overset{U}{\mapsto} \quad$}
\begin{tikzpicture}
[scale = 0.5,
v/.style = {circle, fill = black, inner sep = 0.6mm},
u/.style = {circle, fill = white, inner sep = 0.1mm}
]
\node[u] (15) at (3, -2) {\textcolor{blue}{{\small$\frac{2}{\deg x} - 1$}}};
\node[u] (12) at (-0.7, 1.25) {\textcolor{blue}{{\small$\frac{2}{\deg x}$}}};
\node[u] (13) at (1.3, 2.8) {\textcolor{blue}{{\small$\frac{2}{\deg x}$}}};
\node[u] (14) at (1.8, -0.4) {\textcolor{blue}{{\small$\frac{2}{\deg x}$}}};
\node[u] (x) at (0.7, 0) {$x$};
\node[v] (1) at (0.5,0.5) {};
\node[v] (2) at (0.5, 2) {};
\node[v] (3) at (2, 2) {};
\node[v] (4) at (2, 0.5) {};
\draw[line width = 1pt] (1) to (2);
\draw[line width = 1pt] (1) to (3);
\draw[line width = 1pt] (1) to (4);
\node[v] (5) at (-0.5,-0.5) {};
\node[v] (6) at (-0.5, -2) {};
\node[v] (7) at (-2, -2) {};
\node[v] (8) at (-2, -0.5) {};
\draw[line width = 1pt] (1) to (5);
\draw[line width = 1pt] (5) to (6);
\draw[line width = 1pt] (5) to (7);
\draw[line width = 1pt] (5) to (8);
\node[u] (15R) at (0.5, 0.2) {};
\node[u] (15L) at (-0.2, -0.5) {};
\draw[draw = blue, line width = 1pt, ->] (15R) to (15L);
\node[u] (12o) at (0.2, 0.7) {};
\node[u] (12t) at (0.2, 1.8) {};
\draw[draw = blue, line width = 1pt, ->] (12o) to (12t);
\node[u] (13o) at (0.65, 0.93) {};
\node[u] (13t) at (1.7, 2) {};
\draw[draw = blue, line width = 1pt, ->] (13o) to (13t);
\node[u] (14o) at (0.7, 0.3) {};
\node[u] (14t) at (1.8, 0.3) {};
\draw[draw = blue, line width = 1pt, ->] (14o) to (14t);
\draw[-, blue] (1.5,-1.9) to [bend left = 15] (0.3,-0.3);
\end{tikzpicture}
\caption{The action of $U$: Each blue arrow represents the standard basis vector labeled by the corresponding arc. 
The left and right figures represent the before and after states of the action $U$, respectively. We omit depicting arrows of corresponding arcs returning the value $0$ of the states.} \label{0519-3}
\end{center}
\end{figure}
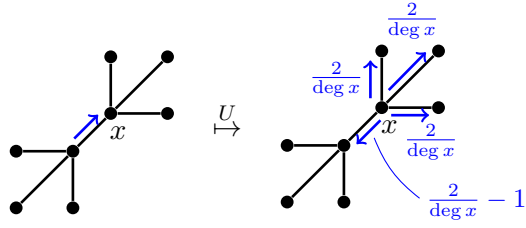

Similarly, the action of $U^*$ on the standard basis vector $\BM{e}_a$ is given as follows.
The proof is similar to that of Lemma~\ref{0507-1},
and hence is omitted.

\begin{lem} \label{0513-1}
For $a \in \MC{A}$, we have
\[ U^* \BM{e}_a = \frac{2}{\deg o(a)} \sum_{\substack{z \in \MC{A} \\ t(z) = o(a) }} \BM{e}_z - \BM{e}_{a^{-1}}. \]
\end{lem}

The action described in Lemma~\ref{0513-1} can also be visualized as shown in Figure~\ref{0602-1}.
In this figure, the action of $U^*$ collects the arcs ending at the vertex $y$,
where $y$ is the initial vertex of the original arc.
Except for the reverse arc of the original one,
these arcs have weight $\frac{2}{\deg y}$,
while the reverse arc has weight $\frac{2}{\deg y} - 1$.

\begin{figure}[htb]
\begin{center}
\begin{tikzpicture}
[scale = 0.5,
v/.style = {circle, fill = black, inner sep = 0.6mm},
u/.style = {circle, fill = white, inner sep = 0.1mm}
]
\node[v] (1) at (0.5,0.5) {};
\node[v] (2) at (0.5, 2) {};
\node[v] (3) at (2, 2) {};
\node[v] (4) at (2, 0.5) {};
\draw[line width = 1pt] (1) to (2);
\draw[line width = 1pt] (1) to (3);
\draw[line width = 1pt] (1) to (4);
\node[v] (5) at (-0.5,-0.5) {};
\node[v] (6) at (-0.5, -2) {};
\node[v] (7) at (-2, -2) {};
\node[v] (8) at (-2, -0.5) {};
\draw[line width = 1pt] (1) to (5);
\draw[line width = 1pt] (5) to (6);
\draw[line width = 1pt] (5) to (7);
\draw[line width = 1pt] (5) to (8);
\node[u] (51L) at (-0.5, -0.2) {};
\node[u] (51R) at (0.2, 0.5) {};
\draw[draw = blue, line width = 1pt, ->] (51L) to (51R);
\node[u] (x) at (0.2, -0.7) {$y$};
\node[inner sep=0pt, outer sep=0pt, minimum size=0pt] (d) at (0,-3.5) {};
\end{tikzpicture}
\raisebox{45pt}{$\quad \overset{U^*}{\mapsto} \quad$}
\begin{tikzpicture}
[scale = 0.5,
v/.style = {circle, fill = black, inner sep = 0.6mm},
u/.style = {circle, fill = white, inner sep = 0.1mm}
]
\node[u] (15) at (3.5, -0.5) {\textcolor{blue}{{\small$\frac{2}{\deg y} - 1$}}};
\node[u] (12) at (0.7, -1.5) {\textcolor{blue}{{\small$\frac{2}{\deg y}$}}};
\node[u] (13) at (-1.3, -2.8) {\textcolor{blue}{{\small$\frac{2}{\deg y}$}}};
\node[u] (14) at (-1.5, 0.5) {\textcolor{blue}{{\small$\frac{2}{\deg y}$}}};
\node[u] (x) at (0.2, -0.7) {$y$};
\node[v] (1) at (0.5,0.5) {};
\node[v] (2) at (0.5, 2) {};
\node[v] (3) at (2, 2) {};
\node[v] (4) at (2, 0.5) {};
\draw[line width = 1pt] (1) to (2);
\draw[line width = 1pt] (1) to (3);
\draw[line width = 1pt] (1) to (4);
\node[v] (5) at (-0.5,-0.5) {};
\node[v] (6) at (-0.5, -2) {};
\node[v] (7) at (-2, -2) {};
\node[v] (8) at (-2, -0.5) {};
\draw[line width = 1pt] (1) to (5);
\draw[line width = 1pt] (5) to (6);
\draw[line width = 1pt] (5) to (7);
\draw[line width = 1pt] (5) to (8);
\node[u] (15R) at (0.5, 0.2) {};
\node[u] (15L) at (-0.2, -0.5) {};
\draw[draw = blue, line width = 1pt, ->] (15R) to (15L);
\node[u] (12o) at (-0.2, -0.7) {};
\node[u] (12t) at (-0.2, -1.8) {};
\draw[draw = blue, line width = 1pt, ->] (12t) to (12o);
\node[u] (13o) at (-0.65, -0.93) {};
\node[u] (13t) at (-1.7, -2) {};
\draw[draw = blue, line width = 1pt, ->] (13t) to (13o);
\node[u] (14o) at (-0.7, -0.3) {};
\node[u] (14t) at (-1.8, -0.3) {};
\draw[draw = blue, line width = 1pt, ->] (14t) to (14o);
\draw[-, blue] (2, -0.3) to [bend left = -15] (0.3,-0.2);
\end{tikzpicture}
\caption{The action of $U^*$: The action $U^*$ is given by flipping the direction of every arrow in before and after states of the action $U$.} \label{0602-1}
\end{center}
\end{figure}
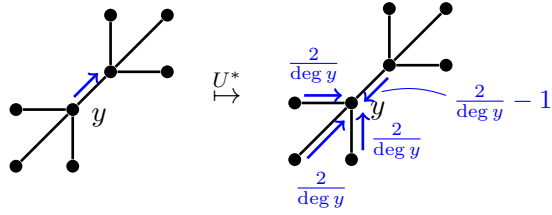

\subsection{Additional material on algebraic graph theory}

Let $G = (V, E)$ be a graph.
A mapping $g: V \to V$ is an \emph{automorphism} of $G$
if $g$ is bijective, and $\{x,y\} \in E$
if and only if $\{g(x), g(y) \} \in E$.
We denote the set of all automorphisms of $G$ by $\Aut(G)$.


Let $g$ be an automorphism of a graph $G$,
and let $\MC{A} = \MC{A}(G)$ be the symmetric arc set.
Define $\tilde{g}: \MC{A} \to \MC{A}$ by
$\tilde{g}(\left(x,y)\right) = (g(x), g(y))$.
Clearly, $\tilde{g}$ is a bijection.
Thus, the permutation matrix $N_{\tilde{g}} \in \MB{C}^{\MC{A} \times \MC{A}}$ is defined by $(N_{\tilde{g}})_{a,b} = \delta_{a, \tilde{g}(b)}$.
Note that
\begin{equation} \label{0703-1}
N_{\tilde{g}} \BM{e}_a = \BM{e}_{\tilde{g}(a)}
\end{equation}
for any $a \in \MC{A}$.
The matrix $N_{\tilde{g}}$ commutes with the time evolution operator $U$ defined in Section~\ref{0519-1}:


\begin{lem}[{\cite[Lemma~3.2]{kubota2025circulant}}] \label{0519-2}
Let $g$ be an automorphism of a graph $G$,
and let $U = U(G)$ be the time evolution operator of the Grover walk on $G$.
Then we have $U N_{\tilde{g}} = N_{\tilde{g}} U$.
\end{lem}

Let $G$ and $H$ be graphs,
and let $A$ and $B$ be their adjacency matrices, respectively.
We call the graph whose adjacency matrix is $A \otimes B$ the \emph{Kronecker product} of $G$ and $H$,
and denote it by $G \otimes H$.
Equivalently,
$G \otimes H$ is the graph whose vertex set is $V(G) \times V(H)$,
where two vertices $(x_1, y_1)$ and $(x_2, y_2)$ are adjacent if and only if $\{x_1, x_2\} \in E(G)$ and $\{y_1, y_2\} \in E(H)$.
In this paper, we consider the Kronecker product of a graph with itself.
We briefly discuss the symmetry of such graphs.

\begin{lem} \label{0501-1}
Let $G$ be a graph, and let $g:V(G \otimes G) \to V(G \otimes G)$ be a mapping defined by $g(x,y) = (y,x)$.
Then $g$ is an automorphism of $G \otimes G$.
\end{lem}

\begin{proof}
Clearly, the mapping $g$ is bijective.
We show that $g$ preserves the adjacency relation.
Indeed,
\begin{align*}
\{ (x_1, y_1), (x_2, y_2) \} \in E(G \otimes G)
&\iff \{x_1, x_2\}, \{y_1, y_2\} \in E(G) \\
&\iff \{y_1, y_2\}, \{x_1, x_2\} \in E(G) \\
&\iff \{ (y_1, x_1), (y_2, x_2) \} \in E(G \otimes G) \\
&\iff \{ g(x_1, y_1), g(x_2, y_2) \} \in E(G \otimes G).
\end{align*}
Hence, $g$ is an automorphism of $G \otimes G$.
\end{proof}

\section{Two-particle quantum walks} \label{0527-1}

Let $G$ be a graph.
We call the Hilbert space $\MC{H}^{(2)}(G) := \MB{C}^{\MC{A}(G)} \otimes \MB{C}^{\MC{A}(G)}$ with the standard inner product the \emph{two-particle Hilbert space on $G$}. 
When the graph $G$ under consideration is clear from the context,
we simply write $\MC{H}^{(2)}$.
For $\psi \in \MC{H}^{(2)}$, we denote by $\psi(a,b)$ the coefficient of $\BM{e}_a \otimes \BM{e}_b$.
That is, we write
\begin{equation} \label{0611-1}
\psi = \sum_{a,b \in \MC{A}} \psi(a,b) \BM{e}_a \otimes \BM{e}_b.
\end{equation}
Define $P: \MC{H}^{(2)} \to \MC{H}^{(2)}$ by $P(\BM{e}_a \otimes \BM{e}_b) = \BM{e}_b \otimes \BM{e}_a$ and extend it linearly.
We call this linear operator the \emph{swap operator}.
For convenience, we use the same symbol $P$ for the matrix representation of this linear operator with respect to the standard basis $\{\BM{e}_a \otimes \BM{e}_b \mid a,b \in \MC{A}(G) \}$ of $\MC{H}^{(2)}$.
A vector $\psi \in \MC{H}^{(2)}$ is said to be a \emph{bosonic state} if $P\psi = \psi$, and a \emph{fermionic state} if $P\psi = -\psi$.
Our first aim is to find a natural correspondence between the two-particle Hilbert space on $G$ and the one-particle Hilbert space on another graph.

\subsection{Two-particle quantum walk via the one-particle Grover walk on $G \otimes G$}

The most naive way to define the time evolution operator of a two-particle quantum walk on $\MC{H}^{(2)}$ is to use $U \otimes U$,
where $U$ is the time evolution operator of the one-particle quantum walk on each tensor factor $\MB{C}^{\MC{A}(G)}$.
However, this time evolution is nothing but applying the time evolution operator of the one-particle quantum walk independently to each particle,
and no nontrivial behavior as a two-particle system can be expected.
At least for the entanglement entropy defined later,
it does not change under this time evolution.
Hence, we consider a one-particle quantum walk on the Kronecker product of graphs.

\begin{lem} \label{0422-1}
Let $G$ be a graph, and define a mapping $r: \MC{A}(G) \times \MC{A}(G) \to \MC{A}(G \otimes G)$ by
\[ r((x,y), (z,w)) = ((x,z), (y,w)) \]
for $(x,y), (z,w) \in \MC{A}(G)$.
Then, $r$ is well-defined and bijective.
\end{lem}

\begin{proof}
Once it is shown that $r$ is well-defined,
it is clear that $r$ is bijective.
Therefore, we show that $r$ is well-defined, that is,
we show that if $(x,y), (z,w) \in \MC{A}(G)$,
then $((x,z), (y,w)) \in \MC{A}(G \otimes G)$.
Indeed,
\begin{align*}
(x,y), (z,w) \in \MC{A}(G)
&\iff \{x,y\}, \{z,w\} \in E(G) \\
&\iff \{ (x,z), (y,w) \} \in E(G \otimes G) \\
&\iff ((x,z), (y,w)) \in \MC{A}(G \otimes G),
\end{align*}
as claimed.
\end{proof}

The bijection $r$ given in Lemma~\ref{0422-1} allows us to identify the two-particle Hilbert space on $G$ with the one-particle Hilbert space on $G \otimes G$.
Indeed, let $R: \MC{H}^{(2)}(G) \to \MB{C}^{\MC{A}(G \otimes G)}$ be the linear operator obtained by linearly extending the correspondence $R(\BM{e}_a \otimes \BM{e}_b) = \BM{e}_{r(a,b)}$.
Then, $R$ is an isomorphism between $\MC{H}^{(2)}(G)$ and $\MB{C}^{\MC{A}(G \otimes G)}$.
The action of $R$ is important in actual computations.
In particular, when $a=(x,y)$ and $b=(z,w)$, we have
\[
R(\BM{e}_{(x,y)} \otimes \BM{e}_{(z,w)}) = \BM{e}_{r((x,y),(z,w))} = \BM{e}_{(x,z),(y,w)}.
\]
Here and throughout, $\BM{e}_{(x,z),(y,w)}$ denotes $\BM{e}_{((x,z),(y,w))}$.
It is convenient to remember the action of $R$ as removing ``$\otimes \BM{e}$" and interchanging $y$ and $z$.

For a graph $G$,
let $U = U(G \otimes G)$ be the time evolution operator of the Grover walk on $G \otimes G$, and define the operator $\MC{U} = \MC{U}(G)$ on $\MC{H}^{(2)}(G)$ by
\begin{equation} \label{0423-1}
\MC{U} := R^{-1}UR.
\end{equation}
Since $R$ and $U$ are unitary, $\MC{U}$ is also unitary.
We call $\MC{U}$ the \emph{time evolution operator of the two-particle Grover walk on $G$}.
That is, we define the \emph{two-particle Grover walk on $G$} via the one-particle Grover walk on $G \otimes G$.

\begin{pro}
Let $\MC{U}$ be the time evolution operator of the two-particle Grover walk on a graph $G$.
Then, we have
\[ \MC{U}(\BM{e}_{(x,y)} \otimes \BM{e}_{(z,w)})
= \frac{2}{\deg_G y \, \deg_G w} \sum_{\substack{r,s \in V(G) \\ \{r,y\}, \{s,w\} \in E(G)}} \BM{e}_{(y,r)} \otimes \BM{e}_{(w,s)} - \BM{e}_{(y,x)} \otimes \BM{e}_{(w,z)} \]
for $(x,y), (z,w) \in \MC{A}(G)$.
\end{pro}

\begin{proof}
We recall that
\begin{align*}
\MC{A}(G \otimes G) &= \left \{ ( (p,q),(r,s) ) \, \middle | \, \begin{aligned} &p,q,r,s \in V(G), \\ &\{(p,q),(r,s)\} \in E(G \otimes G) \end{aligned} \right \} \\
&= \left \{ ( (p,q),(r,s) ) \, \middle | \, \begin{aligned} &p,q,r,s \in V(G), \\ &\{p,r\}, \{q,s\} \in E(G) \end{aligned} \right \}.
\end{align*}
Then, we have
\begin{align*}
\MC{U}(\BM{e}_{(x,y)} \otimes \BM{e}_{(z,w)})
&= R^{-1}UR(\BM{e}_{(x,y)} \otimes \BM{e}_{(z,w)}) \\
&= R^{-1}U \BM{e}_{(x,z), (y,w)} \\
&= R^{-1} \left\{ \frac{2}{\deg_{G \otimes G}(y,w)} \sum_{\substack{\alpha \in \MC{A}(G \otimes G) \\ o(\alpha) = (y,w)}} \BM{e}_{\alpha} - \BM{e}_{(y,w), (x,z)} \right\} \\
&= R^{-1} \left\{ \frac{2}{\deg_G y \, \deg_G w} \sum_{\substack{p,q,r,s \in V(G) \\ \{p,r\}, \{q,s\} \in E(G) \\ (p,q) = (y,w)}} \BM{e}_{(p,q),(r,s)} - \BM{e}_{(y,w), (x,z)} \right\} \\
&= R^{-1} \left\{ \frac{2}{\deg_G y \, \deg_G w} \sum_{\substack{r,s \in V(G) \\ \{r,y\}, \{s,w\} \in E(G)}} \BM{e}_{(y,w),(r,s)} - \BM{e}_{(y,w), (x,z)} \right\} \\
&= \frac{2}{\deg_G y \, \deg_G w} \sum_{\substack{r,s \in V(G) \\ \{r,y\}, \{s,w\} \in E(G)}} \BM{e}_{(y,r)} \otimes \BM{e}_{(w,s)} - \BM{e}_{(y,x)} \otimes \BM{e}_{(w,z)},
\end{align*}
as claimed.
\end{proof}

For this two-particle walk to be consistent with quantum mechanics,
the time evolution operator $\MC{U}$ is required to commute with the swap operator $P$~\cite{sakurai2020modern}.
This commutativity can be shown by considering the symmetry of $G \otimes G$.

\begin{thm} \label{0527-2}
Let $\MC{U}$ be the time evolution operator of the two-particle Grover walk on a graph $G$,
and let $P$ be the swap operator.
Then, we have $\MC{U}P = P\MC{U}$.
\end{thm}

\begin{proof}
By Lemma~\ref{0501-1},
the mapping $g:V(G \otimes G) \to V(G \otimes G)$ defined by $g(x,y) = (y,x)$ is an automorphism of $G \otimes G$.
We consider the permutation matrix $N_{\tilde{g}}$.
By~\eqref{0703-1}, we have
\[ N_{\tilde{g}}\BM{e}_{(x,y),(z,w)} = \BM{e}_{\tilde{g}((x,y), (z,w))}
= \BM{e}_{g(x,y), g(z,w)} = \BM{e}_{(y,x), (w,z)} \]
for any $((x,y), (z,w)) \in \MC{A}(G \otimes G)$.
On the other hand,
\[ RPR^{-1}\BM{e}_{(x,y),(z,w)}
= RP(\BM{e}_{(x,z)} \otimes \BM{e}_{(y,w)})
= R (\BM{e}_{(y,w)} \otimes \BM{e}_{(x,z)}) = \BM{e}_{(y,x), (w,z)}. \]
Thus, we have
\begin{equation}
N_{\tilde{g}} = RPR^{-1},
\end{equation}
and hence Lemma~\ref{0519-2} yields
\begin{align*}
\MC{U}P &= (R^{-1}UR)(R^{-1}N_{\tilde{g}}R) \\
&= R^{-1}UN_{\tilde{g}}R \\
&= R^{-1}N_{\tilde{g}}UR \\
&= R^{-1}(RPR^{-1})(R\MC{U}R^{-1})R \\
&= P\MC{U}
\end{align*}
as claimed.
\end{proof}

\subsection{Entanglement entropy} \label{0602-3}

Let $G$ be a graph,
and let $\MC{H}^{(2)}(G)$ be the two-particle Hilbert space on $G$.
For a quantum state $\psi \in \MC{H}^{(2)}$, 
we use the same notation as in~\eqref{0611-1} and write
\[ \psi = \sum_{a,b \in \MC{A}} \psi(a,b) \BM{e}_a \otimes \BM{e}_b. \]
We call the matrix $\Psi \in \MB{C}^{\MC{A}(G) \times \MC{A}(G)}$ defined by
\begin{equation} \label{0504-1}
\Psi_{a,b} := \psi(a,b)
\end{equation}
the \emph{amplitude matrix} of $\psi$.
For any $\nu \in \mathbb{C}^{(\MC{A}(G) \times \MC{A}(G)) \times (\MC{A}(G) \times \MC{A}(G))}$, the partial traces $\tr_1(\nu), \tr_2(\nu) \in \mathbb{C}^{\MC{A}(G) \times \MC{A}(G)}$ are defined by 
\[ (\tr_1(\nu))_{a,b}=\sum_{c \in \MC{A}(G)} \nu_{(c,a),(c,b)}, \quad (\tr_2(\nu))_{a, b}=\sum_{c\in \MC{A}(G)} \nu_{(a,c),(b,c)} \]
for $ a, b \in \MC{A}(G)$.
Set $\rho_j=\tr_j(\psi\psi^*)$ for $\psi \in \mathcal{H}^{(2)}(G)$ ($j=1,2$). 
We define the quantity $S(\psi)$ for $\psi\in \mathcal{H}^{(2)}(G)$ by 
\[S(\psi) := -\tr(\rho_2 \log \rho_2), \]
where $\log$ denotes the natural logarithm and $\rho_2 \log \rho_2$ is understood through the spectral decomposition of $\rho_2$ with the convention $0 \log 0 = 0$.
Since $\rho_2=\Psi\Psi^*$ and $\rho_1=\Phi\Phi^*$, where $\Phi=\Psi^\top$, the value $S(\psi)$ is independent of the choice of the subsystem traced out, that is, 
\[ S(\psi)=-\tr(\rho_j\log \rho_j)\;\;(j=1,2). \]
Then, 
although the particle labels are not operationally accessible in a system of indistinguishable particles, the value $S(\psi)$ is independent of the choice of the subsystem traced out.
In this paper, we use $S(\psi)$ as the entanglement entropy.

The following alternative expression for $S(\psi)$ will be useful to compute its value concretely:
\begin{equation} \label{0428-1}
S(\psi) = - \sum_{\lambda \in \Spec(\Psi \Psi^*)} \lambda \log \lambda,
\end{equation}
where $\Spec(\Psi \Psi^*)$ denotes the multiset of eigenvalues of $\Psi \Psi^*$.
Thus, the quantity $S(\psi)$ is the \emph{von Neumann entropy} of the one-particle reduced density matrix.
In this paper, we refer to $S(\psi)$ as the \emph{entanglement entropy}.
Note that, for identical particles, several inequivalent notions of entanglement have been proposed in the literature, including approaches based on the one-particle reduced density matrix, mode entanglement, operational definitions, and detector-level entanglement~\cite{benatti2020entanglement, ghirardi2004general, pavskauskas2001quantum,  tichy2013entanglement, wiseman2003entanglement}.

Here, we briefly check the properties of the eigenvalues of $\Psi \Psi^*$.
First, since $\Psi \Psi^*$ is a Gram matrix, it is positive semidefinite.
In particular, all of its eigenvalues are nonnegative.
Furthermore, the sum of the eigenvalues of $\Psi \Psi^*$ is equal to $1$.
Indeed,
\begin{align*}
\tr(\Psi \Psi^*) &= \sum_{z \in \MC{A}(G)}(\Psi \Psi^*)_{z,z}
= \sum_{z \in \MC{A}(G)} \sum_{a \in \MC{A}(G)} \Psi_{z,a} (\Psi^*)_{a,z} \\
&= \sum_{z \in \MC{A}(G)} \sum_{a \in \MC{A}(G)} \Psi_{z,a} \overline{\Psi_{z,a}}
= \sum_{z,a \in \MC{A}(G)} \psi(z,a) \overline{\psi(z,a)}
= \| \psi \|^2 = 1.
\end{align*}
Thus, the eigenvalues of $\Psi\Psi^*$ can be regarded as a probability distribution.

Let us compute the entanglement entropy of a quantum state in the two-particle Hilbert space $\MC{H}^{(2)}(G)$ on a graph $G$ having at least one edge $\{x,y\}$.
Consider the following state,
which will be used as an initial state in this paper:
\[ \psi = \frac{1}{\sqrt{2}} \left\{ \BM{e}_{(x,y)} \otimes \BM{e}_{(y,x)} \pm \BM{e}_{(y,x)} \otimes \BM{e}_{(x,y)} \right\}. \]
In the case of $+$, the state is bosonic,
while in the case of $-$, it is fermionic.
For either choice of sign,
with respect to a suitable ordering of $\MC{A}(G)$,
we have
\[ 
\Psi = 
\begin{bmatrix} 
0&1/\sqrt{2}&0&\cdots&0 \\ 
\pm1/\sqrt{2}&0&0&\cdots&0 \\ 
0&0&0&\cdots&0 \\ 
\vdots&\vdots&\vdots&\ddots&\vdots \\ 
0&0&0&\cdots&0
\end{bmatrix},
\]
and hence $\Psi \Psi^* = \diag(1/2, 1/2, 0, \dots, 0)$.
Therefore, the entanglement entropy $S(\psi)$ is given by
\begin{equation} \label{0602-2}
S(\psi) = - \frac{1}{2} \log \frac{1}{2} - \frac{1}{2} \log \frac{1}{2} = \log 2.
\end{equation}

We note that, for identical particles, the von Neumann entropy $S(\psi)$ may contain a contribution arising solely from exchange symmetry.
In this paper, the value $S(\psi)=\log 2$ simply refers to the entropy defined in~\eqref{0428-1}, and we do not interpret this value as operationally accessible entanglement.

For two-particle quantum walks,
the following bounds on the entanglement entropy of quantum states are known.

\begin{lem}[{\cite[Theorem~11.8]{nielsen2010quantum}}]
Let $\psi$ be a quantum state in the two-particle Hilbert space $\MC{H}^{(2)}(G)$.
Then, we have
\[ 0 \leq S(\psi) \leq \log |\MC{A}(G)|, \]
and equality in the upper bound holds if and only if all eigenvalues $\lambda$ of $\Psi \Psi^*$ satisfy
\[ \lambda = \frac{1}{|\MC{A}(G)|}. \]
\end{lem}

Since $\Psi \Psi^*$ is Hermitian,
its diagonalization allows us to reformulate the condition for attaining the upper bound of the entanglement entropy as follows.

\begin{cor} \label{0507-2}
Let $\psi$ be a quantum state in the two-particle Hilbert space $\MC{H}^{(2)}(G)$, and let $\Psi$ be the amplitude matrix of $\psi$. Then, $S(\psi)=\log |\MC{A}(G)|$ if and only if
\[ \Psi \Psi^* = \frac{1}{|\MC{A}(G)|} I. \]
\end{cor}

We are interested in graphs for which our two-particle Grover walk
produces a state whose entanglement entropy attains the upper bound.

\section{Two-particle Grover walks on complete bipartite graphs $K_{n,n}$}

In this section, we study the entanglement entropy
for two-particle Grover walks on complete bipartite graphs.
In general,
when a two-particle Grover walk evolves from a given initial state,
it is not necessarily easy to describe explicitly the amplitude matrix $\Psi$ defined by~\eqref{0504-1} at each time step. 
Accordingly, in this paper,
we consider two-particle Grover walks on complete bipartite graphs
as a trial example and study their entanglement entropy.
As we will see later, the Kronecker product of a complete bipartite graph with itself is the disjoint union of two complete bipartite graphs.
Furthermore, since the one-particle Grover walk on a complete bipartite graph is periodic, our two-particle Grover walk is also periodic.
Taking this fact into account,
we can compute the entanglement entropy explicitly
by directly calculating the time evolution at several specific times.

\begin{lem} \label{0505-1}
For the complete bipartite graph $K_{n,n}$,
the Kronecker product of $K_{n,n}$ with itself
is isomorphic to the disjoint union of two copies of $K_{n^2, n^2}$.
That is,
\[ K_{n,n} \otimes K_{n,n} \simeq K_{n^2, n^2} \cup K_{n^2, n^2}. \]
\end{lem}

\begin{proof}
Although a direct proof based on the definition of the adjacency relation of the Kronecker product is possible, it is rather involved,
so we prove the statement via the eigenvalues of the adjacency matrix.
See Chapter~1 of~\cite{brouwer2011spectra} for basic facts concerning graph spectra.
Since $\Spec(K_{n,n}) = \{ [n]^{1}, [0]^{2n-2}, [-n]^{1} \}$,
we have
\[ \Spec(K_{n,n} \otimes K_{n,n}) =  \left \{ [n^2]^{2}, [0]^{4n^2 - 4}, [-n^2]^{2} \right \}. \]
Moreover, the graph $K_{n,n} \otimes K_{n,n}$ is $n^2$-regular.
Since the multiplicity of the eigenvalue $n^2$ is $2$,
it follows that $K_{n,n} \otimes K_{n,n}$ has two connected components, say $G_1$ and $G_2$.
Each component is connected and $n^2$-regular,
and hence has $n^2$ as an eigenvalue with multiplicity one.
On the other hand, since the eigenvalues of the adjacency matrix of each component sum to $0$,
each component must also have $-n^2$ as an eigenvalue with multiplicity one,
and all its remaining eigenvalues are $0$.
In particular, $\rank(A(G_1)) = \rank(A(G_2)) = 2$.
Therefore, both $G_1$ and $G_2$ must be complete bipartite graphs~\cite[Lemma~3.2]{cheng2007nullity}.
Suppose that $G_1 = K_{p,q}$.
Since $G_1$ is $n^2$-regular, we obtain $p=q=n^2$.
The same holds for $G_2$.
Hence, we see that $K_{n,n} \otimes K_{n,n}$ is isomorphic to $K_{n^2, n^2} \cup K_{n^2, n^2}$.
\end{proof}

By the above lemma,
the two-particle Grover walk on the complete bipartite graph $K_{n,n}$ evolves, in effect, governed by the one-particle Grover walk on $K_{n^2, n^2}$.
For the computations below,
we fix the following labeling of $K_{n,n}$:
\begin{align*}
V(K_{n,n}) &:= \{ x_1, \dots, x_n, y_1, \dots, y_n \}, \\
E(K_{n,n}) &:= \{ \{x_i, y_j\} \mid i,j \in \{1, \dots,n\} \}.
\end{align*}
We note that the valency of $K_{n^2, n^2}$ is $n^2$
and that $|\MC{A}(K_{n,n})| = 2n^2$.

Let $U = U(K_{n,n} \otimes K_{n,n})$ be the time evolution operator of the one-particle Grover walk.
Following~\eqref{0423-1},
we define the time evolution operator $\MC{U} = \MC{U}(K_{n,n})$ of the two-particle Grover walk on $K_{n,n}$ by $\MC{U} := R^{-1}UR$.
Using the edge $\{x_1, y_1\}$, define the initial states
\begin{equation} \label{0507-4}
\psi_0^{\pm} := \frac{1}{\sqrt{2}} \left\{ \BM{e}_{(x_1,y_1)} \otimes \BM{e}_{(y_1,x_1)} \pm \BM{e}_{(y_1,x_1)} \otimes \BM{e}_{(x_1,y_1)} \right \}.
\end{equation}
For each choice of sign, we define $\psi_{\tau}^{\pm} := {\MC{U}}^{\tau} \psi_0^{\pm}$ for a non-negative integer $\tau$.
Theorem~1.3 in~\cite{higuchi2017periodicity}, together with Lemma~\ref{0505-1}, shows that the one-particle Grover walk on a complete bipartite graph with at least three vertices is $4$-periodic, that is, $U^4 = I$.
This implies that the two-particle Grover walk on $K_{n,n}$ is also $4$-periodic for $n \geq 2$, that is, $\MC{U}^4 = I$.
Therefore, it is sufficient to compute $\psi_{\tau}^{\pm}$ for $\tau = 0,1,2,3$.
For these states,
we denote the corresponding amplitude matrices by $\Psi_{\tau}^{\pm}$.
First, we compute the states up to time $2$.

\begin{lem} \label{0520-1}
We have
\begin{align*}
\psi_1^{\pm} &= 
\frac{\sqrt{2}}{n^2} \left\{
\sum_{i,j=1}^n \BM{e}_{(y_1, x_i)} \otimes \BM{e}_{(x_1, y_j)}
\pm \sum_{i,j=1}^n \BM{e}_{(x_1, y_i)} \otimes \BM{e}_{(y_1, x_j)} \right\} \\
& \hspace{8em}
- \frac{1}{\sqrt{2}} \BM{e}_{(y_1, x_1)} \otimes \BM{e}_{(x_1, y_1)}
\mp \frac{1}{\sqrt{2}} \BM{e}_{(x_1, y_1)} \otimes \BM{e}_{(y_1, x_1)},
\intertext{and}
\psi_2^{\pm}
&= \frac{\sqrt{2}}{n^2} \Biggl(
\frac{2}{n^2} \sum_{i,j=1}^n \sum_{k, \ell = 1}^{n} \BM{e}_{(x_i, y_k)} \otimes \BM{e}_{(y_j, x_{\ell})} - \sum_{i,j=1}^n \BM{e}_{(x_i, y_1)} \otimes \BM{e}_{(y_j, x_1)} \\
& \hspace{2em} \pm \frac{2}{n^2} \sum_{i,j=1}^n \sum_{k, \ell = 1}^{n} \BM{e}_{(y_i, x_k)} \otimes \BM{e}_{(x_j, y_{\ell})} \mp \sum_{i,j=1}^n \BM{e}_{(y_i, x_1)} \otimes \BM{e}_{(x_j, y_1)} \Biggr) \\
& \hspace{4em} \mp \frac{\sqrt{2}}{n^2} \left(
\sum_{i,j=1}^n \BM{e}_{(y_1, x_i)} \otimes \BM{e}_{(x_1, y_j)}
\pm \sum_{i,j=1}^n \BM{e}_{(x_1, y_i)} \otimes \BM{e}_{(y_1, x_j)} \right) \\
& \hspace{6em} +\frac{1}{\sqrt{2}} \BM{e}_{(x_1, y_1)} \otimes \BM{e}_{(y_1, x_1)} \pm \frac{1}{\sqrt{2}} \BM{e}_{(y_1, x_1)} \otimes \BM{e}_{(x_1, y_1)}.
\end{align*}
\end{lem}

\begin{proof}
First, we compute the state $\psi_1^{\pm}$.
We have
\[ \psi_1^{\pm} = \MC{U} \psi_0^{\pm} = R^{-1}UR \psi_0^{\pm}
= R^{-1}U \frac{1}{\sqrt{2}} \left\{ \BM{e}_{(x_1,y_1), (y_1,x_1)} \pm \BM{e}_{(y_1,x_1), (x_1,y_1)} \right \}. \]
Thus, it follows from Lemma~\ref{0507-1} that
\begin{align}
R \psi_1^{\pm}
&= U \frac{1}{\sqrt{2}} \left\{ \BM{e}_{(x_1,y_1), (y_1,x_1)} \pm \BM{e}_{(y_1,x_1), (x_1,y_1)} \right \} \\
&= \frac{1}{\sqrt{2}} \left\{ \frac{2}{n^2} \sum_{i,j=1}^n \BM{e}_{(y_1, x_1),(x_i, y_j)} - \BM{e}_{(y_1, x_1),(x_1, y_1)} \pm \frac{2}{n^2} \sum_{i,j=1}^n \BM{e}_{(x_1, y_1),(y_i, x_j)} \mp \BM{e}_{(x_1, y_1),(y_1, x_1)}
\right \} \notag \\
&= \frac{\sqrt{2}}{n^2} \left\{
\sum_{i,j=1}^n \BM{e}_{(y_1, x_1),(x_i, y_j)}
\pm \sum_{i,j=1}^n \BM{e}_{(x_1, y_1),(y_i, x_j)} \right\}
- \frac{1}{\sqrt{2}} \BM{e}_{(y_1, x_1),(x_1, y_1)}
\mp \frac{1}{\sqrt{2}} \BM{e}_{(x_1, y_1),(y_1, x_1)}. \label{0507-3}
\end{align}
Hence, the assertion follows for the state $\psi_1^{\pm}$.

Next, we compute the state $\psi_2^{\pm}$.
For this purpose, it is convenient to rewrite $R \psi_1^{\pm}$ in~\eqref{0507-3} as
\[ R\psi_1^{\pm} = 
\frac{\sqrt{2}}{n^2} \left(
\sum_{i,j=1}^n \BM{e}_{(y_1, x_1), (x_i, y_j)}
\pm \sum_{i,j=1}^n \BM{e}_{(x_1, y_1), (y_i, x_j)} \right) \mp R \psi_0^{\pm}. \]
We have
\begin{align*}
R\psi_2^{\pm}
&= \frac{\sqrt{2}}{n^2} \left(
\sum_{i,j=1}^n U \BM{e}_{(y_1, x_1), (x_i, y_j)}
\pm \sum_{i,j=1}^n U \BM{e}_{(x_1, y_1), (y_i, x_j)} \right) \mp R\psi_1^{\pm} \\
&= \frac{\sqrt{2}}{n^2} \Biggl\{
\sum_{i,j=1}^n \left( \frac{2}{n^2} \sum_{k, \ell = 1}^{n} \BM{e}_{(x_i, y_j), (y_k, x_{\ell})} - \BM{e}_{(x_i, y_j), (y_1, x_1)}
\right) \\
& \hspace{4em} \pm \sum_{i,j=1}^n \left( \frac{2}{n^2} \sum_{k, \ell = 1}^{n} \BM{e}_{(y_i, x_j), (x_k, y_{\ell})} - \BM{e}_{(y_i, x_j), (x_1, y_1)}
\right)
\Biggr\} \mp R\psi_1^{\pm} \\
&= \frac{\sqrt{2}}{n^2} \Biggl(
\frac{2}{n^2} \sum_{i,j=1}^n \sum_{k, \ell = 1}^{n} \BM{e}_{(x_i, y_j), (y_k, x_{\ell})} - \sum_{i,j=1}^n \BM{e}_{(x_i, y_j), (y_1, x_1)} \\
& \hspace{2em} \pm \frac{2}{n^2} \sum_{i,j=1}^n \sum_{k, \ell = 1}^{n} \BM{e}_{(y_i, x_j), (x_k, y_{\ell})} \mp \sum_{i,j=1}^n \BM{e}_{(y_i, x_j), (x_1, y_1)} \Biggr) \\
& \hspace{4em} \mp \frac{\sqrt{2}}{n^2} \left(
\sum_{i,j=1}^n \BM{e}_{(y_1, x_1),(x_i, y_j)}
\pm \sum_{i,j=1}^n \BM{e}_{(x_1, y_1),(y_i, x_j)} \right) \\
& \hspace{6em} + \frac{1}{\sqrt{2}} \BM{e}_{(x_1, y_1),(y_1, x_1)} \pm
\frac{1}{\sqrt{2}} \BM{e}_{(y_1, x_1),(x_1, y_1)}.
\end{align*}
Hence, the assertion follows for the state $\psi_2^{\pm}$.
\end{proof}

Next, we deal with the exceptional case $n=1$.

\begin{lem} \label{0520-2}
For the two-particle Grover walk on $K_{1,1}$ starting from the initial states defined in~\eqref{0507-4}, the entanglement entropy attains its upper bound at every time.
\end{lem}

\begin{proof}
First, we prove by induction on $\tau$ that
\[ \psi_{\tau}^{\pm} = (\pm 1)^{\tau} \psi_{0}^{\pm} \]
holds for any time $\tau$.
Indeed, the assertion is trivial when $\tau = 0$.
We assume that the equality holds up to time $\tau - 1$.
By Lemma~\ref{0520-1}, we can verify that
\begin{align*}
\psi_{\tau}^{\pm} &= \MC{U} \psi_{\tau - 1}^{\pm}
= \MC{U} (\pm 1)^{\tau-1} \psi_{0}^{\pm} = (\pm 1)^{\tau-1} \psi_{1}^{\pm} \\
&= (\pm 1)^{\tau-1} \Big(
\sqrt{2} \BM{e}_{(y_1, x_1)} \otimes \BM{e}_{(x_1,y_1)}
\pm \sqrt{2} \BM{e}_{(x_1, y_1)} \otimes \BM{e}_{(y_1,x_1)} \\
& \hspace{8em} - \frac{1}{\sqrt{2}} \BM{e}_{(y_1, x_1)} \otimes \BM{e}_{(x_1,y_1)}
\mp \frac{1}{\sqrt{2}} \BM{e}_{(x_1, y_1)} \otimes \BM{e}_{(y_1,x_1)}
\Big) \\
&= (\pm 1)^{\tau-1} \frac{1}{\sqrt{2}} \left(
\BM{e}_{(y_1, x_1)} \otimes \BM{e}_{(x_1,y_1)}
\pm \BM{e}_{(x_1, y_1)} \otimes \BM{e}_{(y_1,x_1)} \right) \\
&= (\pm 1)^{\tau} \psi_{0}^{\pm}.
%
\end{align*}
Thus, the corresponding amplitude matrix is
\[ \Psi_{\tau}^{\pm} =
\begin{bmatrix} 
0&(\pm1)^{\tau}/\sqrt{2} \\ 
(\pm1)^{\tau+1}/\sqrt{2}&0
\end{bmatrix} \]
and hence we have
\[ \Psi_{\tau}^{\pm} (\Psi_{\tau}^{\pm})^*
= \begin{bmatrix} 
1/2&0 \\ 
0&1/2
\end{bmatrix}
= \frac{1}{2} I = \frac{1}{|\MC{A}(K_{1,1})|}I. \]
By Corollary~\ref{0507-2},
the entanglement entropy attains its upper bound at every time $\tau$.
\end{proof}

For $n \neq 1$,
the only case in which the entanglement entropy attains its upper bound at some time is $n=2$.
As shown below, if $n \neq 1,2$,
then the entanglement entropy does not attain its upper bound at any time.

\begin{lem} \label{0520-3}
For the two-particle Grover walk on $K_{2,2}$ starting from the initial states defined in~\eqref{0507-4}, the entanglement entropy attains its upper bound at time $\tau \equiv 2 \pmod{4}$.
\end{lem}

\begin{proof}
With respect to the ordering
\[
(x_1,y_1), (x_1,y_2), (x_2,y_1), (x_2,y_2),
(y_1,x_1), (y_1,x_2), (y_2,x_1), (y_2,x_2)
\]
of $\MC{A}(K_{2,2})$,
the amplitude matrix corresponding to the state $\psi_2^{\pm}$ obtained in Lemma~\ref{0520-1} is given by
\[
\Psi_2^{\pm} = \frac{1}{4\sqrt{2}}
\begin{bmatrix}
0&0&0&0&1&-1&-1&1\\
0&0&0&0&-1&-1&1&1\\
0&0&0&0&-1&1&-1&1\\
0&0&0&0&1&1&1&1\\
\pm 1&\mp 1&\mp 1&\pm 1&0&0&0&0\\
\mp 1&\mp 1&\pm 1&\pm 1&0&0&0&0\\
\mp 1&\pm 1&\mp 1&\pm 1&0&0&0&0\\
\pm 1&\pm 1&\pm 1&\pm 1&0&0&0&0
\end{bmatrix}.
\]
A direct calculation shows that
\[
\Psi_2^{\pm}(\Psi_2^{\pm})^* = \frac{1}{8}I = \frac{1}{|\MC{A}(K_{2,2})|}I,
\]
and hence Corollary~\ref{0507-2} implies that the entanglement entropy attains its upper bound at time $\tau = 2$.
Since the complete bipartite graph is $4$-periodic,
the assertion of the lemma follows.
\end{proof}

Our second main result is that, for complete bipartite graphs,
the cases in which the entanglement entropy attains its upper bound are precisely those described above.

\begin{thm} \label{0527-3}
For the two-particle Grover walk on $K_{n,n}$ starting from the initial states defined in~\eqref{0507-4},
the entanglement entropy attains its upper bound at some time if and only if $n = 1,2$.
Moreover, when $n=1$, the entanglement entropy attains its upper bound at every time, while when $n=2$, it attains its upper bound at time $\tau$ if and only if $\tau \equiv 2 \pmod{4}$.
\end{thm}

\begin{proof}
The sufficiency part has already been proved by Lemmas~\ref{0520-2} and~\ref{0520-3},
so it remains to prove the necessity.
Let $n \geq 2$.
First, for the initial states $\psi_0^{\pm}$,
the entanglement entropy does not attain its upper bound.
This has essentially already been verified in Section~\ref{0602-3}.
Indeed, it follows from \eqref{0602-2} that
\[ S(\psi_0^{\pm}) = \log 2 \neq \log (2n^2) = \log |\MC{A}(K_{n,n})|. \]

The states at times $\tau \in \{1,2 \}$, namely $\psi_1^{\pm}$ and $\psi_2^{\pm}$, have already been computed in Lemma~\ref{0520-1},
and hence the corresponding amplitude matrices can also be obtained.
The $((x_1, y_1), (x_1, y_1))$-entry of $\Psi_{1}^{\pm} (\Psi_{1}^{\pm})^*$ is
\begin{align*}
(\Psi_{1}^{\pm} (\Psi_{1}^{\pm})^*)_{(x_1, y_1), (x_1, y_1)}
&= \sum_{z \in \MC{A}(K_{n,n})} |\psi_{1}^{\pm}((x_1, y_1), z)|^2 \\
&= \left( \frac{\sqrt{2}}{n^2} \right)^2 \cdot (n-1)
+ \left( \frac{\sqrt{2}}{n^2} - \frac{1}{\sqrt{2}}  \right)^2 \\
&= \frac{n^3 - 4n + 4}{2n^3}.
\end{align*}
However, this value is not equal to $\frac{1}{2n^2}$ for any $n \geq 2$.
Hence Corollary~\ref{0507-2} implies that the entanglement entropy of $\psi_1^{\pm}$ does not attain its upper bound.

Next, the $((x_1, y_1), (x_1, y_1))$-entry of
$\Psi_{2}^{\pm} (\Psi_{2}^{\pm})^*$ is
\begin{align*}
(\Psi_{2}^{\pm} (\Psi_{2}^{\pm})^*)_{(x_1, y_1), (x_1, y_1)}
&= \sum_{z \in \MC{A}(K_{n,n})} |\psi_{2}^{\pm}((x_1, y_1), z)|^2 \\
&= \sum_{i,j=1}^n |\psi_{2}^{\pm}((x_1, y_1), (y_i, x_j))|^2 \\
&= \sum_{i,j = 2}^n |\psi_{2}^{\pm}((x_1, y_1), (y_i, x_j))|^2
+ \sum_{j=2}^n |\psi_{2}^{\pm}((x_1, y_1), (y_1, x_j))|^2 \\
& \hspace{2em} + \sum_{i=2}^n |\psi_{2}^{\pm}((x_1, y_1), (y_i, x_1))|^2
+ |\psi_{2}^{\pm}((x_1, y_1), (y_1, x_1))|^2 \\
&= \left( \frac{2 \sqrt{2}}{n^4} \right)^2 \cdot (n-1)^2 + \left( \frac{2 \sqrt{2}}{n^4} - \frac{\sqrt{2}}{n^2} \right)^2 \cdot (n-1) \\
& \hspace{2em} + \left\{ \frac{\sqrt{2}}{n^2} \left( \frac{2}{n^2} - 1 \right) \right\}^2 \cdot (n-1) + \left\{ \frac{\sqrt{2}}{n^2} \left( \frac{2}{n^2} - 1 \right) - \frac{\sqrt{2}}{n^2} + \frac{1}{\sqrt{2}} \right\}^2 \\
&= \frac{8(n-1)^2}{n^8} + \frac{4(n^2-2)^2 (n-1)}{n^8} + \frac{(n^2-2)^4}{2n^8} \\
&= \frac{(n^3-4n+4)^2}{2n^6}.
\end{align*}
By Corollary~\ref{0507-2},
if the entanglement entropy attains its upper bound,
then $\Psi_2^{\pm}(\Psi_2^{\pm})^* = \frac{1}{2n^2}I$ must hold.
Thus, solving
\[ \frac{(n^3-4n+4)^2}{2n^6} = \frac{1}{2n^2}, \]
in positive integers $n \geq 2$, we have $n=2$.

Continuing, we compute the states $\psi_3^{\pm}$ at time $3$.
Since the two-particle Grover walk on a complete bipartite graph is $4$-periodic,
we have $\psi_3^{\pm} = \mathcal{U}^3 \psi_0^{\pm} = \mathcal{U}^{-1} \psi_0^{\pm}$.
Therefore, it is convenient to use the inverse of the time evolution operator.
Since
\[ \psi_3^{\pm} = \mathcal{U}^{-1} \psi_0^{\pm} = R^{-1} U^* R \psi_0^{\pm} = R^{-1} U^* \frac{1}{\sqrt{2}} \left\{ \BM{e}_{(x_1,y_1), (y_1,x_1)} \pm \BM{e}_{(y_1,x_1), (x_1,y_1)} \right \}, \]
it follows from Lemma~\ref{0513-1} that
\begin{align*}
R \psi_3^{\pm} &= U^* \frac{1}{\sqrt{2}} \left\{ \BM{e}_{(x_1,y_1), (y_1,x_1)} \pm \BM{e}_{(y_1,x_1), (x_1,y_1)} \right \} \\
&= \frac{1}{\sqrt{2}} \left\{ \frac{2}{n^2} \sum_{i,j=1}^n \BM{e}_{(y_i, x_j), (x_1, y_1)} - \BM{e}_{(y_1, x_1), (x_1, y_1)}
\pm \frac{2}{n^2} \sum_{i,j=1}^n \BM{e}_{(x_i,y_j),(y_1,x_1)} \mp \BM{e}_{(x_1, y_1),(y_1,x_1)} \right\}.
\end{align*}
Hence,
\begin{align*}
\psi_3^{\pm} &= \frac{\sqrt{2}}{n^2} \left\{
\sum_{i,j=1}^n \BM{e}_{(y_i, x_1)} \otimes \BM{e}_{(x_j, y_1)}
\pm \sum_{i,j=1}^n \BM{e}_{(x_i,y_1)} \otimes \BM{e}_{(y_j,x_1)}
\right\} \\
& \hspace{8em}
- \frac{1}{\sqrt{2}} \BM{e}_{(y_1, x_1)} \otimes \BM{e}_{(x_1, y_1)}
\mp \frac{1}{\sqrt{2}} \BM{e}_{(x_1, y_1)} \otimes \BM{e}_{(y_1,x_1)}.
\end{align*}
Therefore, the $((x_1, y_1), (x_1, y_1))$-entry of
$\Psi_{3}^{\pm} (\Psi_{3}^{\pm})^*$ is
\begin{align*}
(\Psi_{3}^{\pm} (\Psi_{3}^{\pm})^*)_{(x_1, y_1), (x_1, y_1)}
&= \sum_{z \in \MC{A}(K_{n,n})} |\psi_{3}^{\pm}((x_1, y_1), z)|^2 \\
&= \left( \frac{\sqrt{2}}{n^2} \right)^2 \cdot (n-1)
+ \left( \frac{\sqrt{2}}{n^2} - \frac{1}{\sqrt{2}}  \right)^2 \\
&= \frac{n^3 - 4n + 4}{2n^3}.
\end{align*}
However, this value is not equal to $\frac{1}{2n^2}$ for any $n \geq 2$.
\end{proof}

In fact, the reduced density matrices $\Psi_{\tau}^{\pm}(\Psi_{\tau}^{\pm})^*$ corresponding to the two choices of sign coincide, that is,
\[ \Psi_{\tau}^{+}(\Psi_{\tau}^{+})^* = \Psi_{\tau}^{-}(\Psi_{\tau}^{-})^* \]
holds for any time $\tau$.
We verify this equality only for the case where $\tau$ is even.
First, we have
\[
\psi_{\tau}^{\pm} = \mathcal{U}^{\tau} \psi_{0}^{\pm}
= \mathcal{U}^{\tau} \left\{ \frac{1}{\sqrt{2}} \bm{e}_{(x_1, y_1)} \otimes \bm{e}_{(y_1, x_1)} \right\} \pm \mathcal{U}^{\tau} \left\{ \frac{1}{\sqrt{2}} \bm{e}_{(y_1, x_1)} \otimes \bm{e}_{(x_1, y_1)} \right\}.
\]
Since $\tau$ is even, the bipartiteness of $K_{n,n}$ implies that
\[ \MC{U}^{\tau} \left\{ \frac{1}{\sqrt{2}} \BM{e}_{(x_1, y_1)} \otimes \BM{e}_{(y_1, x_1)} \right\} \]
and
\[ \MC{U}^{\tau} \left\{ \frac{1}{\sqrt{2}} \BM{e}_{(y_1, x_1)} \otimes \BM{e}_{(x_1, y_1)} \right\} \]
have amplitude matrices of the forms
\[ \begin{bmatrix} O&A_{\tau} \\ O&O \end{bmatrix}, \qquad \begin{bmatrix} O&O \\ B_{\tau}&O \end{bmatrix}, \]
respectively,  with respect to a suitable ordering of $\mathcal A(K_{n,n})$.
Therefore, the amplitude matrix of $\psi_{\tau}^{\pm}$ is given by
\[ \Psi_{\tau}^{\pm} = \begin{bmatrix} O&A_{\tau} \\ \pm B_{\tau}&O \end{bmatrix}, \]
and hence
\[ \Psi_{\tau}^{\pm}(\Psi_{\tau}^{\pm})^*
= \begin{bmatrix} O&A_{\tau} \\ \pm B_{\tau}&O \end{bmatrix} \begin{bmatrix} O&\pm B_{\tau}^* \\ A_{\tau}^*&O \end{bmatrix}
= \begin{bmatrix} A_{\tau}A_{\tau}^*&O \\ O&B_{\tau}B_{\tau}^* \end{bmatrix}, \]
which is independent of the choice of sign.
The case where $\tau$ is odd can be treated in the same manner.

\section{Summary and further discussion}

In this paper, we defined a two-particle quantum walk on a graph $G$ via the one-particle Grover walk on the Kronecker product $G \otimes G$,
and proved that its time evolution operator $\MC{U}$ commutes with the swap operator $P$ using methods from algebraic graph theory.
Moreover, for the two-particle Grover walk on the complete bipartite graph $K_{n,n}$ starting from the initial states defined in~\eqref{0507-4}, we showed that the entanglement entropy attains its upper bound at some time if and only if $n=1,2$.
Our approach focused on the fact that the Grover walks on complete bipartite graphs are $4$-periodic,
and explicitly computed the states up to time $3$.
Although this approach works well in the case of $K_{n,n}$,
it seems difficult to determine whether the entanglement entropy attains its upper bound by such direct computations for graphs other than $K_{n,n}$, even if the corresponding Grover walk is periodic.
Thus, a natural next problem would be to find an analytic method with some degree of generality,
even if the graph does not induce a periodic Grover walk.
For example, one possible direction is to study whether the attainment of the upper bound of the entropy can be studied by using the eigenvalues of matrices associated with the graph or those of the time evolution operator.
Another possible problem is to obtain useful necessary conditions for the entropy to attain its upper bound.
It is also natural to study how the structure of the graph is related to the entropy.

In the early stage of this study,
our numerical experiments showed that not only $K_{1,1}$ and $K_{2,2}$,
but also the path graph $P_5$ on five vertices exhibits the attainment of the upper bound of the entanglement entropy for the same initial states.
Interestingly, the graphs $K_{1,1}$, $K_{2,2}$, and $P_5$ are also known to exhibit periodicity~\cite{higuchi2017periodicity, kubota2021periodicity} and perfect state transfer~\cite{kubota2022perfect}.
In studies of periodicity and perfect state transfer,
it is known that twice certain eigenvalues of the discriminant, 
which is a normalized adjacency matrix, are algebraic integers~\cite[Figure~1]{kubota2025strongly}.
We are interested in whether an analogous statement also holds for entanglement entropy.

\section*{Acknowledgements}
We would like to thank the anonymous reviewers for their valuable comments, suggestions, and references, which helped improve the quality of the paper.
S.K. is supported by JSPS KAKENHI (Grant Number JP24K16970).
E.S. is supported by JSPS KAKENHI (Grant Number JP24K06863).

\section*{Data availability}
This article has no associated data.

\bibliographystyle{plain}
\bibliography{mybibs}

@article{ahlbrecht2012molecular,
  title={Molecular binding in interacting quantum walks},
  author={Ahlbrecht, A. and Alberti, A. and Meschede, D. and Scholz, V. B. and Werner, A. H. and Werner, R. F.},
  journal={New Journal of Physics},
  volume={14},
  number={7},
  pages={073050},
  year={2012},
  publisher={IOP Publishing}
}

@article{berry2011two,
  title={Two-particle quantum walks: Entanglement and graph isomorphism testing},
  author={Berry, S. D. and Wang, J. B.},
  journal={Physical Review A—Atomic, Molecular, and Optical Physics},
  volume={83},
  number={4},
  pages={042317},
  year={2011},
  publisher={APS}
}

@book{brouwer2011spectra,
  title={Spectra of graphs},
  author={Brouwer, A. E. and Haemers, W. H.},
  year={2011},
  publisher={Springer Science \& Business Media}
}

@article{carson2015entanglement,
  title={Entanglement dynamics of two-particle quantum walks},
  author={Carson, G. R. and Loke, T. and Wang, J. B.},
  journal={Quantum Information Processing},
  volume={14},
  number={9},
  pages={3193--3210},
  year={2015},
  publisher={Springer}
}

@article{camacho2025quantum,
  title={Quantum walk on a square lattice with identical particles},
  author={Camacho, G. and Meinecke, J. and Wolters, J.},
  journal={Physical Review A},
  volume={111},
  number={5},
  pages={052416},
  year={2025},
  publisher={APS}
}

@article{chandrashekar2012quantum,
  title={Quantum walk on distinguishable non-interacting many-particles and indistinguishable two-particle},
  author={Chandrashekar, C.M. and Busch, T.},
  journal={Quantum Information Processing},
  volume={11},
  number={5},
  pages={1287--1299},
  year={2012},
  publisher={Springer}
}

@article{cheng2007nullity,
  title={On the nullity of graphs},
  author={Cheng, B. and Liu, B.},
  journal={The Electronic Journal of Linear Algebra},
  volume={16},
  pages={60--67},
  year={2007}
}

@article{gamble2010two,
  title={Two-particle quantum walks applied to the graph isomorphism problem},
  author={Gamble, J. K. and Friesen, M. and Zhou, D. and Joynt, R. and Coppersmith, S. N.},
  journal={arXiv preprint arXiv:1002.3003},
  year={2010}
}

@book{godsil2013algebraic,
  title={Algebraic graph theory},
  author={Godsil, C. and Royle, G.},
  volume={207},
  year={2013},
  publisher={Springer Science \& Business Media}
}

@article{goyal2010spatial,
  title={Spatial entanglement using a quantum walk on a many-body system},
  author={Goyal, S. K. and Chandrashekar, C. M.},
  journal={Journal of Physics A: Mathematical and Theoretical},
  volume={43},
  number={23},
  pages={235303},
  year={2010}
}

@article{higuchi2017periodicity,
  title={Periodicity of the discrete-time quantum walk on a finite graph},
  author={Higuchi, Y. and Konno, N. and Sato, I. and Segawa, E.},
  journal={Interdisciplinary Information Sciences},
  volume={23},
  number={1},
  pages={75--86},
  year={2017},
  publisher={The Editorial Committee of the Interdisciplinary Information Sciences}
}

@article{kubota2021periodicity,
  title={Periodicity of quantum walks defined by mixed paths and mixed cycles},
  author={Kubota, S. and Sekido, H. and Yata, H.},
  journal={Linear Algebra and its Applications},
  volume={630},
  pages={15--38},
  year={2021},
  publisher={Elsevier}
}

@article{kubota2021quantum,
  title={Quantum walks defined by digraphs and generalized Hermitian adjacency matrices},
  author={Kubota, S. and Segawa, E. and Taniguchi, T.},
  journal={Quantum Information Processing},
  volume={20},
  number={3},
  pages={95},
  year={2021},
  publisher={Springer}
}

@article{kubota2022perfect,
  title={Perfect state transfer in Grover walks between states associated to vertices of a graph},
  author={Kubota, S. and Segawa, E.},
  journal={Linear Algebra and its Applications},
  volume={646},
  pages={238--251},
  year={2022},
  publisher={Elsevier}
}

@article{kubota2025circulant,
  title={Circulant graphs with valency up to 4 that admit perfect state transfer in Grover walks},
  author={Kubota, S. and Yoshino, K.},
  journal={Journal of Combinatorial Theory, Series A},
  volume={216},
  pages={106064},
  year={2025},
  publisher={Elsevier}
}

@article{kubota2025strongly,
  title={Strongly regular and strongly walk-regular graphs that admit perfect state transfer},
  author={Kubota, S. and Sekido, H. and Yata, H. and Yoshino, K.},
  journal={arXiv preprint arXiv:2506.02530},
  year={2025}
}

@article{lovett2010universal,
  title={Universal quantum computation using the discrete-time quantum walk},
  author={Lovett, N. B. and Cooper, S. and Everitt, M. and Trevers, M. and Kendon, V.},
  journal={Physical Review A—Atomic, Molecular, and Optical Physics},
  volume={81},
  number={4},
  pages={042330},
  year={2010},
  publisher={APS}
}

@article{mastandrea2026spatial,
  title={Spatial and particle-particle entanglement in 1D quantum walks of two distinguishable or indistinguishable bosonic particles},
  author={Mastandrea, C. and Chien, C.-C.},
  journal={arXiv preprint arXiv:2606.02505},
  year={2026}
}

@book{nielsen2010quantum,
  title={Quantum computation and quantum information},
  author={Nielsen, M. A. and Chuang, I. L.},
  year={2010},
  publisher={Cambridge university press}
}

@article{omar2006quantum,
  title={Quantum walk on a line with two entangled particles},
  author={Omar, Y. and Paunkovi{\'c}, N. and Sheridan, L. and Bose, S.},
  journal={Physical Review A—Atomic, Molecular, and Optical Physics},
  volume={74},
  number={4},
  pages={042304},
  year={2006},
  publisher={APS}
}

@article{paryzkova2024two,
  title={Two-particle Hadamard walk on dynamically percolated line and circle},
  author={Par{\`y}zkov{\'a}, M. and {\v{S}}tefa{\v{n}}{\'a}k, M. and Novotn{\`y}, J. and Koll{\'a}r, B. and Kiss, T.},
  journal={Physica Scripta},
  volume={99},
  number={3},
  pages={035112},
  year={2024},
  publisher={IOP Publishing}
}

@book{portugal2013quantum,
  title={Quantum walks and search algorithms},
  author={Portugal, R.},
  volume={19},
  year={2013},
  publisher={Springer}
}

@article{qiang2024quantum,
  title={Quantum walk computing: Theory, implementation, and application},
  author={Qiang, X. and Ma, S. and Song, H.},
  journal={Intelligent Computing},
  volume={3},
  pages={0097},
  year={2024},
  publisher={AAAS}
}

@article{rohde2011multi,
  title={Multi-walker discrete time quantum walks on arbitrary graphs, their properties and their photonic implementation},
  author={Rohde, P. P. and Schreiber, A. and {\v{S}}tefa{\v{n}}{\'a}k, M. and Jex, I. and Silberhorn, C.},
  journal={New Journal of Physics},
  volume={13},
  number={1},
  pages={013001},
  year={2011}
}

@article{rohde2013increasing,
  title={Increasing the dimensionality of quantum walks using multiple walkers},
  author={Rohde, P. P. and Schreiber, A. and {\v S}tefa{\v n}{\'a}k, M. and Jex, I. and Gilchrist, A. and Silberhorn, C.},
  journal={Journal of Computational and Theoretical Nanoscience},
  volume={10},
  number={7},
  pages={1644--1652},
  year={2013},
  doi={10.1166/jctn.2013.3104}
}

@article{rudinger2013comparing,
  title={Comparing algorithms for graph isomorphism using discrete-and continuous-time quantum random walks},
  author={Rudinger, K. and Gamble, J. K. and Bach, E. and Friesen, M. and Joynt, R. and Coppersmith, S. N.},
  journal={Journal of Computational and Theoretical Nanoscience},
  volume={10},
  number={7},
  pages={1653--1661},
  year={2013},
  publisher={American Scientific Publishers}
}

@book{sakurai2020modern,
  title={Modern quantum mechanics},
  author={Sakurai, J. J. and Napolitano, J.},
  year={2020},
  publisher={Cambridge university press}
}

@article{wiseman2003entanglement,
  title={Entanglement of indistinguishable particles shared between two parties},
  author={Wiseman, H. M. and Vaccaro, J. A.},
  journal={Physical review letters},
  volume={91},
  number={9},
  pages={097902},
  year={2003},
  publisher={APS}
}

@article{pavskauskas2001quantum,
  title={Quantum correlations in two-boson wave functions},
  author={Pa{\v{s}}kauskas, R. and You, L.},
  journal={Physical Review A},
  volume={64},
  number={4},
  pages={042310},
  year={2001},
  publisher={APS}
}

@article{ghirardi2004general,
  title={General criterion for the entanglement of two indistinguishable particles},
  author={Ghirardi, G. and Marinatto, L.},
  journal={Physical Review A—Atomic, Molecular, and Optical Physics},
  volume={70},
  number={1},
  pages={012109},
  year={2004},
  publisher={APS}
}

@article{tichy2013entanglement,
  title={Entanglement of identical particles and the detection process},
  author={Tichy, M. C. and de Melo, F. and Ku{\'s}, M. and Mintert, F. and Buchleitner, A.},
  journal={Fortschritte der Physik},
  volume={61},
  number={2-3},
  pages={225--237},
  year={2013},
  publisher={Wiley Online Library}
}

@article{benatti2020entanglement,
  title={Entanglement in indistinguishable particle systems},
  author={Benatti, F. and Floreanini, R. and Franchini, F. and Marzolino, U.},
  journal={Physics Reports},
  volume={878},
  pages={1--27},
  year={2020},
  publisher={Elsevier}
}

\end{document}